\RequirePackage[orthodox,l2tabu]{nag}
\documentclass{llncs}
\usepackage[utf8]{inputenc}
\usepackage[T1]{fontenc}
\usepackage{zi4}
\usepackage[english]{babel}
\usepackage{cite}

\usepackage{mathtools}
\usepackage{amsfonts}
\usepackage{amssymb}

\usepackage{booktabs}
\usepackage{cleveref}
\usepackage{url}
\usepackage{tikz}
\usepackage{csquotes}

\newcommand{\Q}[1]{\textsf{P#1}}

\usepackage{logicpuzzle}

\newcommand*{\0}{\textsf{O}}
\newcommand*{\X}{\textsf{X}}

\newcommand*{\V}{}
\newcommand*{\ixe}{\X{} \hspace{0.00 cm}} 

\renewcommand*{\qed}{~\hfill$\square$}
\begin{document}

	\author{
	    Guillaume Bertholon \and
	    Rémi Géraud--Stewart \and
	    Axel Kugelmann \and 
	    Théo Lenoir \and
	    David Naccache}
    \institute{
               Département d'informatique de l'\'ENS, \'Ecole normale supérieure, \\
               CNRS, PSL Research University, Paris, France \\
               \url{first_name.family_name@ens.fr}
               }
      
    \title{At Most 43 Moves, At Least 29}
    \subtitle{Optimal Strategies and Bounds for Ultimate Tic-Tac-Toe}
              
	\maketitle 
    
    \begin{abstract}
    Ultimate Tic-Tac-Toe is a variant of the well known tic-tac-toe (noughts and crosses) board game. Two players compete to win three aligned \enquote{fields}, each of them being a tic-tac-toe game. Each move determines which field the next player must play in.
    
    We show that there exist a winning strategy for the first player, and therefore that there exist an optimal winning strategy taking at most 43 moves; that the second player can hold on at least 29 rounds; and identify any optimal strategy's first two moves. 
    \end{abstract}

\section{Introduction}

    

Ultimate\footnote{The game is also known as \emph{super} or \emph{meta} tic-tac-toe, we chose to follow what seems to be the most common denomination.} Tic-Tac-Toe (U3T) is a two-player zero-sum game with perfect information. U3T is played on a \emph{board} containing 9 \emph{fields} arranged in a $3\times 3$ grid. The nine fields are indexed from $0$ to $8$ as shown in \Cref{board9}. 
Each field itself is further divided into $3\times 3=9$ \emph{spots}. Spots are also indexed from $0$ to $8$ using the same indexing system. In other words, the board contains 81 spots uniquely identified as $(i,j)$, where $i$ is the field and $j$ the spot.

\begin{figure}[!ht]
\centering
\begin{tikzpicture}
\foreach \i in {1,...,8} {
\draw[gray!50!white] (0.333*\i, 0) to (0.333*\i, 3);
\draw[gray!50!white] (0,0.333*\i) to (3,0.333*\i);
}
\draw[thick] (0, 0) rectangle (3, 3);
\draw[thick] (0, 1) rectangle (3, 2);
\draw[thick] (1, 0) rectangle (2, 3);
\node[scale=2] at (0.5, 2.5) {\textbf{0}};
\node[scale=2] at (1.5, 2.5) {\textbf{1}};
\node[scale=2] at (2.5, 2.5) {\textbf{2}};
\node[scale=2] at (0.5, 1.5) {\textbf{3}};
\node[scale=2] at (1.5, 1.5) {\textbf{4}};
\node[scale=2] at (2.5, 1.5) {\textbf{5}};
\node[scale=2] at (0.5, 0.5) {\textbf{6}};
\node[scale=2] at (1.5, 0.5) {\textbf{7}};
\node[scale=2] at (2.5, 0.5) {\textbf{8}};
\draw[dotted] (3, 2) to (5, 3);
\draw[dotted] (3, 1) to (5, 0);
\begin{scope}[xshift=5cm]
\draw (0, 0) rectangle (3, 3);
\draw (0, 1) rectangle (3, 2);
\draw (1, 0) rectangle (2, 3);
\node[scale=1] at (0.5, 2.5) {$(5, 0)$};
\node[scale=1] at (1.5, 2.5) {$(5, 1)$};
\node[scale=1] at (2.5, 2.5) {$(5, 2)$};
\node[scale=1] at (0.5, 1.5) {$(5, 3)$};
\node[scale=1] at (1.5, 1.5) {$(5, 4)$};
\node[scale=1] at (2.5, 1.5) {$(5, 5)$};
\node[scale=1] at (0.5, 0.5) {$(5, 6)$};
\node[scale=1] at (1.5, 0.5) {$(5, 7)$};
\node[scale=1] at (2.5, 0.5) {$(5, 8)$};
\end{scope}
\end{tikzpicture}
\caption{Indexing for fields and spots. The board has $9$ fields (\emph{left}). Each field contains 9 spots, and we write $(i,j)$ for the $j$-th spot of the $i$-th field (\emph{right}).}\label{board9}
\end{figure}
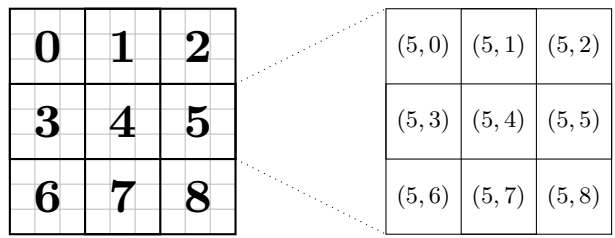

Players alternately place one mark in a free spot, in a field determined according to the rules explained below which we call the \emph{active} field. By convention, the first player (\emph{Xavier}) uses \ixe marks and the second (\emph{Olivia}) uses \0 marks\footnote{The players' initials remind the marks \X,\0.}. 
The active field is constrained by the previous player's move:
\begin{itemize}
    \item If the previous player's move was $(i,j)$, and there are free spots in field $j$, then the active field is $j$;
    \item If the previous player's move was $(i,j)$, and there are no free spots in field $j$, then the current player can freely chose the active field. The same applies if there is no previous move (beginning of the game).
\end{itemize} 
When a player aligns three of their marks (in line, column or diagonal fashion) in a field we say they \emph{win} that field. Any further action in that field will not change this status and is it not possible to win an already won field. When a player aligns three won fields (in line, column, or diagonal fashion) they \emph{win the game}.
It may happen that a field is filled without any player winning it, in which case we say it is a \emph{draw}. Similarly, the full board may be filled without a winner for the game.

By design, there are at most $81$ moves in a game of U3T. \Cref{game:board81} shows an example ongoing game where the first move, by Xavier, was $(4, 4)$ in the center of the board. Next move is Olivia's in the field $0$ (top-left corner of the board).

\begin{figure}[!ht]
\centering
\scalebox{0.7}{
\tikzset{every node/.style={font=\sffamily}}
\begin{lpsudoku}[scale=1]
\setrow{9}{$\0_{11}$,\V,$\X_{14}$,\V,\V,\V,\V,\V,\V}
\setrow{8}{$\0_5$,\V,$\X_{12}$,\V,\V,\V,\V,\V,\V}
\setrow{7}{\V,\V,$\X_2$,\V,\V,\V,\V,$\0_{15}$,\V}
\setrow{6}{$\X_{10}$,\V,\V,$\0_1$,\V,\V,$\0_{13}$,\V,\V}
\setrow{5}{\V,\V,\V,\V,$\X_0$,\V,\V,\V,\V}
\setrow{4}{$\X_6$,\V,\V,\V,\V,\V,\V,\V,\V}
\setrow{3}{\V,\V,\V,$\X_{16}$,\V,\V,$\X_4$,\V,\V}
\setrow{2}{\V,\V,\V,$\0_9$,\V,\V,\V,\V,\V}
\setrow{1}{\V,$\0_7$,\V,\V,$\X_8$,\V,\V,\V,$\0_3$}
\end{lpsudoku}
}
\caption{Example of a U3T board, Xavier has won the field 0. Last move was $\X_{16}$ and it is now Olivia's turn to play, in the active field $0$. Note that although Olivia can align three $\0$ in this field, this field has already been won by Xavier and such a move would therefore have no effect.}\label{game:board81}
\end{figure}

\section{A winning strategy for Xavier}\label{sec:strategy}
In conventional tic-tac-toe, the first player cannot lose if they play perfectly. If the second player also plays perfectly, then the game ends in a draw. In U3T the situation is different and the first player has a simple winning strategy. Such a strategy seems to have been known before\footnote{We found an implementation dating before 2013: \url{https://tinyurl.com/ULTTT}.} but
to the best of the authors' knowledge it is the first time that such a strategy is formally described, analysed, and shown to always succeed. Several results are known about the structure of winning U3T strategies \cite{george2016group}.
Xavier's winning strategy has three phases: an \emph{opening}, a \emph{middlegame}, and an \emph{endgame}. 

\subsection{Opening}
\begin{enumerate}
    \item First move: play $(4, 4)$ (centre spot of the centre field). Olivia is therefore constrained to play in the central field, say $(4, j)$.
    \item Next seven moves: Xavier plays $(j, 4)$ (centre spot of the field $j$). Olivia must again play in the central field. 
    \item After Olivia played her eighth move, the middle field is full as in \Cref{game:win_turn16}, and it is Xavier's turn.
\end{enumerate} 
    
\begin{figure}[!ht]
\centering
\scalebox{0.5}{
\tikzset{every node/.style={font=\sffamily}}
\begin{lpsudoku}[scale=1]
\setrow{9}{\V,\V,\V,\V,\V,\V,\V,\V,\V}
\setrow{8}{\V,\V,\V,\V,\X,\V,\V,\X,\V}
\setrow{7}{\V,\V,\V,\V,\V,\V,\V,\V,\V}
\setrow{6}{\V,\V,\V,\0,\0,\0,\V,\V,\V}
\setrow{5}{\V,\X,\V,\0,\X,\0,\V,\X,\V}
\setrow{4}{\V,\V,\V,\0,\0,\0,\V,\V,\V}
\setrow{3}{\V,\V,\V,\V,\V,\V,\V,\V,\V}
\setrow{2}{\V,\X,\V,\V,\X,\V,\V,\X,\V}
\setrow{1}{\V,\V,\V,\V,\V,\V,\V,\V,\V}
\end{lpsudoku}
}
\caption{Typical board after 16 turns using the winning strategy. It is Xavier's turn.}
\label{game:win_turn16}
\end{figure}

\subsection{Middlegame}
It is now Xavier's turn, and one field hasn't been played yet. We assume that field to be field $0$ (the proof is similar if another field is chosen).
\begin{enumerate}
    \item Xavier plays $(0,0)$, forcing Olivia to play in field $0$:
    \begin{enumerate}
        \item If Olivia plays $(0, k)$ with $k \in \{1, 2, 3, 5, 6, 7\}$ then Xavier plays $(k, 0)$, forcing Olivia back to field $0$;
        \item If Olivia plays $(0, k)$ with $k \in \{4, 8\}$, then Xavier plays $(8, 0)$. Indeed this is possible because field $4$ is full and Xavier can choose $8$ as the active field. If Xavier has already played this position, then $(8,8)$ is played instead, and the \emph{middlegame} phase is over.
    \end{enumerate}
Observe that Xavier cannot be sent twice from the field $0$ to the same field, and all the spots Xavier might want to use are free at the beginning of the strategy. 
\end{enumerate}
The game is not over yet: there is at most one field with three \ixe{}s (field $8$), and Olivia has only played in two fields. Moreover, since there is a finite number of spots in the first field, we know that after at most 17 turns, we are in the situation where Xavier plays $(8, 8)$. Therefore Xavier wins field 8 after at most 17 turns. The situation is analoguous to that of \Cref{game:win_turn25}.
\begin{figure}[!ht]
\centering
\scalebox{0.5}{
\tikzset{every node/.style={font=\sffamily}}
\begin{lpsudoku}[scale=1]
\setrow{9}{\X,\0,\V,\X,\V,\V,\V,\V,\V}
\setrow{8}{\0,\0,\V,\V,\X,\V,\V,\X,\V}
\setrow{7}{\V,\V,\0,\V,\V,\V,\V,\V,\V}
\setrow{6}{\X,\V,\V,\0,\0,\0,\V,\V,\V}
\setrow{5}{\V,\X,\V,\0,\X,\0,\V,\X,\V}
\setrow{4}{\V,\V,\V,\0,\0,\0,\V,\V,\V}
\setrow{3}{\V,\V,\V,\V,\V,\V,\X,\V,\V}
\setrow{2}{\V,\X,\V,\V,\X,\V,\V,\X,\V}
\setrow{1}{\V,\V,\V,\V,\V,\V,\V,\V,\X}
\end{lpsudoku}
}
\caption{Typical board after 25 turns into the winning strategy.}
\label{game:win_turn25}
\end{figure}

\subsection{Endgame}
Now the last step of the strategy: whichever field $k$ is active, when it is Xavier's turn, if $(k, 0)$ can be played then it is played; otherwise $(k, 8)$ is played. Eventually, Xavier wins.

\begin{figure}[!ht]
\centering
\scalebox{0.5}{
\tikzset{every node/.style={font=\sffamily}}
\begin{lpsudoku}[scale=1]
\setrow{9}{\X,\0,\0,\X,\V,\V,\X,\V,\V}
\setrow{8}{\0,\0,\0,\V,\X,\V,\V,\X,\V}
\setrow{7}{\0,\0,\0,\V,\V,\X,\V,\V,\X}
\setrow{6}{\X,\V,\V,\0,\0,\0,\X,\V,\V}
\setrow{5}{\V,\X,\V,\0,\X,\0,\V,\X,\V}
\setrow{4}{\V,\V,\X,\0,\0,\0,\V,\V,\X}
\setrow{3}{\X,\V,\V,\X,\V,\V,\X,\0,\0}
\setrow{2}{\V,\X,\V,\V,\X,\V,\0,\X,\0}
\setrow{1}{\V,\V,\X,\V,\V,\X,\0,\0,\X}
\end{lpsudoku}
}
\caption{The grid after the endgame. Players may keep playing after Xavier wins the game.}\label{game:keepplaying}
\end{figure}

\section{Strategy analysis}

Consider the following properties:
\begin{itemize}
\item\textbf{\Q{1}} For all $i$ in $A:=\{1,2,3,5,6,7\}$ if $(i,0)$ and $(i,8)$ are taken by Xavier, then both $(0,i)$ and $(8,i)$ are taken by Olivia. Note that in this case, Xavier has won field $i$, because he has already taken $(i, 4)$ at the beginning of the game.

\item \textbf{\Q{2}} $\forall i \in A$, if only one of $(i,0),(i,8)$ is taken by Xavier it is $(i,0)$, and only one of $(0,i)$ and $(8,i)$ are taken by Olivia, and if it is in field 8, then Olivia plays in field 0.
\item \textbf{\Q{3}}: $\forall i \in A$, if neither $(i, 0)$ nor $(i, 8)$ are taken by Xavier, then neither $(0, i)$ nor $(8, i)$ are taken by Olivia
\item \textbf{\Q{4}}: Olivia must play in either field 0 or 8, and the field where she plays is not full. Assuming \Q{1}, the fields 0 and 8 cannot be full at the same time before the game ends (if both field 0 and 8 are full, Xavier has won fields 1,2,3,5,6,7 and 8 creating a winning situation with fields 2,5 and 8)
\item \textbf{\Q{5}}: Olivia plays in field 0 if and only if there exists a unique $j \in A$ such that $(8,j)$ is taken by Olivia but $(0,j)$ is not. If Olivia is in field 8 there exist no $j \in A$ such that $(8,j)$ is taken by Olivia but $(0,j)$ is not.

We say that $i$ \emph{verifies} \Q{5} if $i \in A$ and $(8,i)$ is taken by Olivia but $(0,i)$ is not. 
\Q{5} means Olivia plays in field 0 if and only if there exists a unique $i\in A$ such that $i$ verifies \Q{5}, and Olivia plays in field 8 if and only if no $i \in A$ verifies \Q{5}.
\item \textbf{\Q{6}}: Olivia has only played in fields $0, 4, 8$ and Xavier has not played in $(0,i)$ and $(8,i)$ $\forall i \in A$.
\end{itemize}

\begin{lemma}
Properties $\mathsf{P1}$--$\mathsf{P6}$ hold at the beginning of the endgame phase.
\end{lemma}

\begin{proof}
At the beginning of the endgame phase, for every $i \in A$, $(8, i)$ is empty. Moreover, if $(0, i)$ is taken by Olivia, then Xavier has already taken $(i, 0)$. This confirms that \Q{1}, \Q{2} and \Q{3} hold at the beginning of the endgame phase. Similarly, Olivia must play in field 8, which isn't full, thus \Q{4} and \Q{5} hold. Finally, \Q{6} is also verified by design of the middlegame phase. \qed 
\end{proof}

\begin{lemma}
If the properties $\mathsf{P1}$--$\mathsf{P6}$ hold and it is Xavier's turn, then they hold at the next round.
\end{lemma}

\begin{proof}
Assume that the properties are true. We exhaust the possibilities:
\begin{itemize}
\item If Olivia must play in field $0$, and plays in spot $i$, then we have $i \in A$ since spots $0, 4, 8$ are already taken (cf. middlegame). By hypothesis of induction, $(i, 0)$ or $(i, 8)$ is empty.\newline
\begin{itemize}
\item If $(i,0)$ is empty, then Xavier plays in this spot.\newline

Let's prove \Q{1}, \Q{2} and \Q{3}:
The properties hold for all for all $j \in A$ except for $i$ because no concerned spot was taken.
Moreover, by induction with \Q{2}, \Q{3} and \Q{6}, we know that $(0,i),(8,i),(i,0)$ and $(i,8)$ were empty before the round.
At the end of the round, only one spot is taken by Xavier in field $i$, and it is $(i,0)$. $(0,i)$ has been taken by Olivia, and $(8,i)$ is still empty. All in all, \Q{1}, \Q{2} and \Q{3} are verified at the end of the round.\newline

Here is the proof for \Q{4} and \Q{5}:
By induction on \Q{5} there is a unique $j \in A$ that verifies \Q{5} at the beginning of the round. We have $j\not =i$, because we have shown above that $(8,i)$ is empty at the beginning of the round.
Moreover, $i$ does not verify \Q{5} at the end of the round either. Since only the situation of $i$ could have changed during the round regarding \Q{5}, the property is still verified at the end of the round (because only $j$ verifies \Q{5} at this point).
As a corollary, field 0 (where Olivia has to play next round) is not full as $(0,j)$ is empty. This gives \Q{4} at the end of the round.

\medskip 

\item If $(i,0)$ is already taken, by \Q{1}, $(i,8)$ is empty. Indeed, at the beginning of the round, the situation is either the one described in \Q{2} or \Q{3} (because $(0,i)$ is empty), in both of which $(i,8)$ is empty.\\
Thus $i$ verifies \Q{5} at the beginning of the round (and is the only one to do so by definition of \Q{5}). Xavier plays in $(i,8)$.\newline

Since at the end of the step $(0,j),(8,j)$ are taken by Olivia, $(j,0)$ and $(j,8)$ by Xavier we have \Q{1}, \Q{2} and \Q{3} verified for i at the end of the round.
\Q{1}, \Q{2} and \Q{3} remain true for other elements of $A$, because the spots concerned didn't change during the round.\newline

\Q{5} is a simple conclusion from the unicity of j at the last step, i.e. if one $i$ verifies \Q{5},it must be $j$ (otherwise this $i$ would have verified \Q{5} at the last step), and since $j$ does not verify \Q{5}, \Q{5} is proved.

With this established, field 8 is full only if field 0 is also full. If this is the case, the game has ended. Else, \Q{4} is verified.


\end{itemize}

\medskip 

\item If Olivia must play in field 8, and decides to play in spot i with  $i \in A$ (spots 0, 4 and 8 are already taken in this field, so Olivia has to do this choice)

\begin{itemize}
\item If $(i,0)$ is empty then Xavier plays in $(i,0)$.\newline

One can see that if some $j \in A$ verifies \Q{5} at the end of the round, then if $j\neq i$, it verified \Q{5} at the step before, which is contradictory by induction hypothesis.
By induction on \Q{2}, \Q{3} and \Q{6}, at the beginning of the round, $(i,0)$, $(i,8)$, $(0,i)$ and $(8,i)$ are empty, and thus at the end of the round, $(i,0)$ is taken by Xavier, $(8,i)$ is taken by Olivia  and the two others are still empty, thus i verifies \Q{5}, and this proves \Q{5}.
\Q{4} is trivial as $(0,i)$ is empty.\newline

\Q{1}, \Q{2} and \Q{3} hold for every $j\neq i,j\in A$ by induction (because the situation did not change), and also for $i$ (we proved above that we are in the situation of \Q{1}).\newline

\item 
If $(i,0)$ is not free, then $(i,8)$ must be free because \Q{1},\Q{2} and \Q{3} are verified at the beginning of the round. That ensures we are in the situation described by \Q{2} before Olivia plays, that is $(i,0)$ and $(0,i)$ are taken, and $(i,8)$ and $(8,i)$ are free.
According to the strategy, Xavier plays in $(i,8)$.\newline

At the end of the round, \Q{1}, \Q{2} and \Q{3} hold. 
Indeed the situation is the one described in \Q{1} for $i$. Indeed, spots $(i,0)$ and $(i,8)$ are taken by Xavier and spots $(0,i)$ and $(8,i)$ are taken by Olivia.
For $j\neq i$, $j\in A$, nothing changed during the round, so $j$ does not bring a contradiction to \Q{1}, \Q{2} or \Q{3}.\newline

\Q{4} also holds because Olivia was sent in field 8. If the field is full, the game is finished because \Q{1} is verified, and the proof ends. \newline

Note that $i$ does not verify \Q{5} because $(0,i)$ is taken, and by induction on \Q{5}, no other $j\in A$ verifies \Q{5} either. This proves \Q{5} at the end of the round.

\end{itemize}
Finally, in all cases, \Q{6} holds trivially at the end of the round. \qed 
\end{itemize}
\end{proof}

\begin{corollary}
Following the strategy described above, properties $\mathsf{P1}$--$\mathsf{P6}$ hold as long as Xavier does not win the game. If Xavier does not win, we can construct an infinite sequence of rounds. Since the game terminates in at most 81 moves, this is impossible, and therefore Xavier eventually wins the game.
\end{corollary}

\section{Results on an optimal winning strategy}
\subsection{Upper bound}
\begin{lemma}\label{lem:upperbound}
An optimal winning strategy for Xavier takes at most 43 moves.
\end{lemma}
\begin{proof}
We have shown that the strategy described in \Cref{sec:strategy} is winning, we just need to show that it takes at most 43 moves, and therefore that any optimal strategy takes at most 43 moves.

Indeed, at turn 43 only one cross $\X$ is missing in some field $h$ with $h\in A$. 
Thus, if Olivia sent Xavier in field 0 at the end of the opening, either the line composed of field 6,7,8 or the column composed of field 2,5,8 is won by Xavier. Similarly, if Olivia plays every spot before the spot 8 in field 0 during middlegame, and then plays every spot except 5 and 7 in field 8, and then plays $(8,5)$, she loses in 43 turns.
On the other hand, if Olivia sent Xavier in field 1 at the end of the opening, either the line composed of field 0,3,6 or the column composed of field 2,5,8 is won by Xavier.
If Olivia plays every spot before the spot 7 in field 1 during middlegame, and then plays every spot except 6 and 8 in field 8, and then plays $(7,6)$, she loses in 43 turns.
In each case, after 43 turns Xavier has won. \qed 
\end{proof}


\begin{remark}
It seems that our winning strategy hinges on Xavier winning field $4$, which essentially makes the opening sequence the bottleneck.

\end{remark}

\subsection{Lower bound}
\begin{lemma}\label{lem:lbs29}
Olivia can resist Xavier's optimal strategy for at least 29 rounds.
\end{lemma}
We show this by providing an explicit strategy for Olivia, which ensures she does not lose before turn 29. The strategy, which we call \emph{LBS}, is as follows: when Olivia plays in field $j$, she chooses the spot $i$ where she plays in this order:
\begin{enumerate}
\item If Olivia is the first to play in field $j$, she chooses $i=j$.
\item Otherwise, Xavier already has a \ixe in field $j$\footnote{Note that if Xavier was not the first to play in field $j$, then Olivia was, and decided to send Xavier back to field $j$.}.
Olivia chooses $i$ such that there is no \ixe in field $i$ and $(j,i)$ is free.
If no such $i$ exists, she picks $i$ at random among the free spots.
\end{enumerate}
If Olivia is sent in a field that is already full, she plays in a random field $j$ according to the strategy above.

\begin{lemma}\label{lem:lbs18} At turn 18 (that is after Xavier and Olivia play 9 moves), if Olivia plays the LBS strategy, there is one and only one \ixe in each field.
\end{lemma}

\begin{proof}[of \Cref{lem:lbs18}] Suppose that at turn 18, there exists a field $j$ where there are no \ixe\footnote{Note that this claim is the exact contradiction of the assumption of the theorem, because there are 9 \ixe in the grid at turn 18}. Then, Olivia did not play in field $j$ either before turn 18, else she would have sent Xavier in field $j$ at turn 17 or before. It follows that in every field, spot $j$ was free before turn 17. 

Moreover, at turn 17, Xavier had 9 \ixe in 8 fields, so one field has at least 2 \ixe. Let $k$ be such a field. At some point in the game, before turn 17, Olivia sent Xavier in field $k$, but there already was an \ixe in this field. However, Olivia could play in spot $j$ (because it was free for every field at this moment) and send Xavier to a field where there is no \ixe.

That means Olivia did not follow the LBS strategy at some point in the game, so there is a contradiction.
\qed 
\end{proof}
\begin{remark}
This proof can easily be adapted to account for the situation where some fields are full (let $E$ be the set of these fields) and all other fields have exactly the spots of $E$ taken. If there are $\alpha$ full fields, let $\beta = 9-\alpha$ be the number of \enquote{free} fields. Then Olivia can force Xavier to play exactly one \ixe in the $\beta$ \enquote{free} fields for his next $\beta$ turns.
\end{remark}

\begin{proof}[of \Cref{lem:lbs29}]
By \Cref{lem:lbs18}, after 18 turns, there is exactly one \ixe in each field. To win, Xavier needs to win 3 fields, in which at least 3 \ixe are required. That means Xavier has to play at least 6 times after turn 18. As a result, Xavier can not win before turn 29. \qed 
\end{proof}

\subsection{Optimal strategy's first two moves}
We call the nine spots $(i, i)$ on the board \emph{doubles}.
\begin{lemma}\label{lem:avoid}
Xavier's optimal strategy's first move is a double.
\end{lemma}
\begin{proof}
The idea of the proof is that if the first move of Xavier is not a double, Olivia can force Xavier's moves between two fields, which ensures Xavier doesn't win, for long enough. 

Suppose that the first move of Xavier is $(i,j)$ with $i\neq j$. Then Olivia plays in $(j,j)$. 
Now Olivia uses a strategy similar to the one used by Xavier with field $j$ as field $0$ and field $j$ as field $8$: wherever Olivia get sent, if she can play in spot $j$ she does so, otherwise she plays in spot $i$. This process ends when fields $i$ and $j$ are full. 

Since the strategy is essentially identical to Xavier's endgame in the strategy described in \Cref{sec:strategy}, it will not be detailed further here.

There are $16$ \ixe and $16$ \0 played, and Xavier can play wherever he wants. Now Olivia can reuse the LBS which will guarantee that during the $14$ next turns, Xavier will have at most one \ixe in each field except fields $i$ and $j$. 

Thus after $46$ turns, \ixe has not won the game. Therefore by \Cref{lem:upperbound} Xavier's strategy cannot be optimal. \qed 
\end{proof}
An illustration of Olivia's strategy is given in \Cref{game:avoid}. 
\begin{figure}[!hpt]
\centering
\scalebox{0.5}{
\tikzset{every node/.style={font=\sffamily}}
\begin{lpsudoku}[scale=1]
\setrow{9}{\0,\X,\X,\0,\V,\V,\0,\V,\V}
\setrow{8}{\X,\X,\X,\V,\V,\V,\V,\V,\V}
\setrow{7}{\X,\X,\X,\V,\V,\0,\V,\V,\0}
\setrow{6}{\0,\V,\V,\0,\V,\V,\0,\V,\V}
\setrow{5}{\V,\V,\V,\V,\V,\V,\V,\V,\V}
\setrow{4}{\V,\V,\0,\V,\V,\0,\V,\V,\0}
\setrow{3}{\0,\V,\V,\0,\V,\V,\X,\X,\X}
\setrow{2}{\V,\V,\V,\V,\V,\V,\X,\X,\X}
\setrow{1}{\V,\V,\0,\V,\V,\0,\X,\X,\0}
\end{lpsudoku}
}
\caption{Assume that Xavier's first move was $(0,8)_\X$, then by following the strategy of \Cref{lem:avoid} Olivia ensures that Xavier doesn't win the game for at least 46 turns.}\label{game:avoid}
\end{figure}

\begin{lemma}\label{lem:avoid2}
If Olivia first move is $(i,j)$, then Xavier must play $(j,i)$    
\end{lemma}

\begin{proof}
Let $(i,i)$ be Xavier's first move, that Olivia plays $(i,j)$ and that the Xavier's second move is $(j,k)$ with $k\neq i$. In a first time, let's assume $k \neq j$.

Olivia has a strategy to block Xavier in fields $k$ and $i$, similar to the one described in the previous proof. Wherever Olivia get sent, if she can play in spot $k$ she does so otherwise she plays in spot $i$. This process will end with field $i$ and field $j$ full, (the only difference with the previous proof is the case of spot $(k,j)$ and spots $(i,k)$ and $(k,i)$. 

There are $15$ \ixe and $15$ \0 played, and Xavier can play wherever he wants. Now Olivia can reuse the LBS which will guarantee that during the $14$ next turns, Xavier will have at most two \ixe in each field except fields $i$ and $j$. 

Thus after $44$ turns, \ixe has not won the game: by \Cref{lem:upperbound} Xavier's strategy is not optimal. 

Finally, if $k=j$ then Olivia can reuse the strategy described in the previous proof (with $i,j$ as given), ensuring that Xavier doesn't win in less than 46 steps, which shows that Xavier's strategy is not optimal either  \qed 
\end{proof}
\Cref{game:avoid2} shows a typical board after this strategy has been used, for $i=0$ and $j=4$ and $k=8$.

\begin{figure}[!hpt]
\centering
\scalebox{0.5}{
\tikzset{every node/.style={font=\sffamily}}
\begin{lpsudoku}[scale=1]
\setrow{9}{\X,\X,\X,\0,\V,\V,\0,\V,\V}
\setrow{8}{\X,\0,\X,\V,\V,\V,\V,\V,\V}
\setrow{7}{\X,\X,\0,\V,\V,\0,\V,\V,\0}
\setrow{6}{\0,\V,\V,\0,\V,\V,\0,\V,\V}
\setrow{5}{\V,\V,\V,\V,\V,\V,\V,\V,\V}
\setrow{4}{\V,\V,\0,\V,\V,\X,\V,\V,\0}
\setrow{3}{\0,\V,\V,\0,\V,\V,\X,\X,\X}
\setrow{2}{\V,\V,\V,\V,\V,\V,\X,\X,\X}
\setrow{1}{\V,\V,\0,\V,\V,\0,\X,\X,\0}
\end{lpsudoku}
}
\caption{A typical board following the strategy of \Cref{lem:avoid2} if the moves were $(0,0)_\X, (0, 4)_\0, (4, 8)_\X$.}\label{game:avoid2}
\end{figure}

\section{Conclusion}
We described a winning strategy for the first player that wins in at most 43 moves, and a strategy for the second player that doesn't lose for at least 29 moves. Thus unlike the standard game of tic-tac-toe the first player has a substantial advantage: he wins even against a perfect adversary as long as he doesn't make a mistake. Our strategies, for both players, may in turn help design more efficient AI players for this game beating the constructions from \cite{sista2016adversarial,chen2018ai}.

\bibliographystyle{alpha}
\bibliography{u3t.bib}

\end{document}